\documentclass[11pt]{article}

\usepackage[margin=1in]{geometry}
\usepackage[T1]{fontenc}
\usepackage{lmodern}
\usepackage{amsmath,amssymb,amsthm,mathtools}
\usepackage{microtype}
\usepackage[hidelinks]{hyperref}
\hypersetup{pdftitle={Balanced Posted Prices in Stationary Two-Sided Markets},
  pdfauthor={Ilan Reuven Cohen},
  pdfkeywords={Posted prices, two-sided markets, stationary prophet inequalities, online algorithms, XOS valuations}}

\newtheorem{theorem}{Theorem}[section]
\newtheorem{lemma}[theorem]{Lemma}
\newtheorem{corollary}[theorem]{Corollary}
\newtheorem{proposition}[theorem]{Proposition}
\theoremstyle{definition}

\theoremstyle{remark}
\newtheorem{remark}[theorem]{Remark}

\newcommand{\OPT}{\mathrm{OPT}}
\newcommand{\LP}{\mathrm{LP}}
\newcommand{\one}{\mathbf 1}
\newcommand{\R}{\mathbb R}

\title{Balanced Posted Prices in Stationary Two-Sided Markets}
\author{Ilan Reuven Cohen\\
Bar-Ilan University, Israel\\
\texttt{ilan-reuven.cohen@biu.ac.il}}
\date{}

\begin{document}
\pagenumbering{gobble}
\hypersetup{pageanchor=false}
\maketitle
\markright{}
\thispagestyle{empty}

\begin{abstract}
We study an intermediary that buys nonperishing items from arriving sellers and sells them
to arriving buyers, subject to limited storage. Sellers and buyers arrive
according to independent Poisson processes. The intermediary pays and charges
the same fixed item prices, aiming to maximize long-run gains from trade:
buyer value minus seller cost. Our benchmark may choose allocations directly,
has unlimited storage, and must purchase each unit before allocating it.
Its optimal long-run gain rate equals the value of a linear program of average
purchase and allocation rates, with or without knowledge of future arrivals.

For buyers whose bundle values are maxima of nonnegative additive valuations
(XOS), we compute static common prices that guarantee half of this benchmark
for arbitrary finite item capacities, even one storage slot per item,
regardless of how ties are resolved without using future arrivals. The prices
are computable in polynomial time using demand and supporting-value queries.
The construction equalizes itemwise lower bounds on buyer and seller
contributions to gains from trade, accounting for procurement costs. We call
this contribution balance. As a single-item foundation, we also equalize the
rates of purchases and sales selected by the price, before checking inventory
availability. We call this rate balance. For capacity \(C\), it yields the
tight guarantee \(C/(C+1)\).
The two principles parallel proofs of the classical prophet inequality based
on half the expected benchmark and on a median threshold, adapting these
balancing ideas to recurring two-sided markets with limited inventory.
\end{abstract}

\clearpage
\pagenumbering{arabic}
\hypersetup{pageanchor=true}

\section{Introduction}

An online intermediary repeatedly acquires items from arriving sellers and
allocates them to arriving buyers. Its objective is to maximize long-run average gains
from trade, measured as the value generated for buyers minus the costs of
acquiring inventory, subject to finite storage capacity. Every arrival requires an
immediate decision without knowledge of future opportunities: buying too
conservatively can leave valuable demand unserved, while buying too
aggressively can fill capacity with costly inventory; selling to the current
buyer realizes value now but may forgo a better future trade. The central
question is whether the same fixed item prices for buying and selling can
coordinate these decisions while buyers retain their choice of bundles.
For buyers whose bundle values are maxima of nonnegative additive valuations
(\emph{XOS} valuations), we give efficiently computable prices achieving half
of the unlimited-storage benchmark, even with one storage slot per item.

\paragraph*{A motivating application.}
Consider a community reuse service that acquires used electronics from their
owners and makes them available to new users. Its aim is to maximize the value
created for users after accounting for the owners' costs of parting with their
items. Owners arrive with items and minimum acceptable payments, while users may view
products as substitutes or want bundles of products. The service has space for
only a few units of each product type. A high item price attracts supply but
deters moderate-value buyers; a low price attracts demand but can leave
valuable requests unserved when stock runs out. We study a common two-sided price
\(p_j\), paid to a seller from whom item \(j\) is acquired and charged to a
buyer who receives it. Transfers do not enter social gain.

Using the same purchase and sale price gives zero cash margin on completed
units. The service may need \emph{working capital}, money to fund inventory
before resale; Section~\ref{sec:comparison} bounds this requirement.
Our welfare model omits operating costs and intermediary profit.

Our benchmark has unlimited storage and may choose purchases and allocations
directly. It knows the entire realized arrival sequence, but may allocate a
unit only after acquiring it from an earlier seller arrival. Both the
benchmark and the online policy start empty.

We compare our prices with a simpler optimization problem, called a flow
relaxation. A seller type specifies the item offered and its cost, and a
buyer type specifies its values for bundles of items. The relaxation chooses
average purchase and allocation rates, respecting the supply and demand of
each type while ignoring arrival order and inventory limits. In our
nonperishing model, its optimal value equals the best long-run gain of both
online and clairvoyant policies with unlimited storage and direct control of
allocations (Section~\ref{sec:multi-model}). The pricing question is how much
of this value fixed common prices can realize with limited storage while
buyers choose their bundles.

The classical prophet problem supplies two proof templates for our pricing
rules. In that problem, independent nonnegative values are revealed
sequentially, an online decision maker may select at most one, and a prophet
selects the maximum. The classical prophet inequality guarantees the online
decision maker at least one half of the prophet's expected
value~\cite{KS77,KS78}. Samuel-Cahn~\cite{SC84} chooses a median of the prophet's maximum,
balancing the probability that this maximum lies above or below the threshold.
For this single-choice problem, Kleinberg and Weinberg~\cite{KW12} instead use
a threshold equal to half the expected prophet value, balancing the threshold
value against the additional value above it. Both thresholds are one-sided
posted prices that accept the first value clearing the price. We adapt
these two proof templates to a recurring two-sided inventory market.

\paragraph*{Two balancing principles.}
The two pricing principles are easiest to see with one item.
Fix a price \(p\). A completed trade from a seller of cost \(a<p\) to a
buyer of value \(v>p\) creates social gain
\[
v-a=(v-p)+(p-a).
\]
The first term is buyer utility and can be realized only when the item is in
inventory. The second is seller surplus and can be realized only when the
inventory has a free slot. Payments at the common posted price cancel except
for the bounded value of terminal inventory over a finite horizon.

For one item, the first rule chooses the price and randomization for indifferent agents
so that the rates of potential purchases and sales are equal. These
\emph{potential purchase and sale intensities} count arrivals that pass the
price and random acceptance rules, before checking for full or empty
inventory. We call the rule
\emph{rate balance}. When their common rate is positive, inventory
increases and decreases at the same rate whenever those moves are possible.
Each inventory level is therefore equally likely in the long run, so buyers
find stock with the same probability that sellers find free capacity.

The second rule uses the flow relaxation to lower-bound buyer utility per
unit time when an item is in stock and seller surplus per unit time when
there is room for it. We choose
the price to make these two bounds equal, each equal to half of the item's
contribution to the relaxation. We call this \emph{contribution balance}.
An item always has stock or free capacity, possibly both, so at least one
bound applies in every inventory state. Nonnegative utility and surplus
allow negative bounds to be replaced by zero. This yields a one-half
guarantee without calculating inventory probabilities. The realized
utility and surplus can differ; the price equalizes their lower bounds.

We study these principles in two models with stationary arrivals: the rates
and agent types remain the same over time.

\paragraph*{One item with finite inventory.}
Our first model has one homogeneous item and inventory capacity \(C\). Sellers
with different costs and buyers with different values arrive in independent
Poisson streams. The intermediary posts one two-sided price, paying it to
sellers whose costs fall below the price when inventory is not full and
charging it to buyers whose values exceed the price when inventory is nonempty.
This model isolates how one price interacts with the inventory process and with
the available storage capacity.

\paragraph*{Multiple items with combinatorial buyers.}
Our second model has \(m\) heterogeneous items, with an arbitrary finite
capacity for each item. Sellers offer individual items at different costs,
while buyers have XOS valuations over bundles. The intermediary posts an
item-price vector, paying \(p_j\) to an accepted seller of item \(j\) and
charging the same \(p_j\) to a buyer. Each buyer chooses a
bundle from the items currently available that maximizes its utility: the
bundle's value minus the sum of its item prices. The buyer consumes at most
one unit of each selected item. When several bundles give the same maximum
utility, \emph{tie-breaking} chooses among them. Tie-breaking also determines
whether an indifferent seller trades. A \emph{nonanticipatory}
rule may use current and past information but cannot use future arrivals.
Buyer substitution couples the item
inventories, so their states can be correlated even though the arrival streams
are independent.

In both models, the posted prices are static: they remain the same over time.

\paragraph*{Poisson arrivals and discrete time.}
Independent Poisson processes are a standard idealization of agents arriving
at unrelated times with constant rates~\cite[Sections~2.1--2.3]{Gallager11Poisson}.
Here these rates are unaffected by trading decisions. Waiting times are
\emph{memoryless}: elapsed waiting time does not change the distribution
of the remaining wait. Poisson arrivals also satisfy
\emph{Poisson Arrivals See Time Averages} (PASTA)~\cite{Wolff82}:
under a nonanticipatory policy, arrivals see inventory according to
its time-average frequencies (Lemma~\ref{lem:poisson-arrivals}).
Listing arrivals chronologically gives independent arrival types, with
probabilities proportional to their rates. Policies that ignore calendar
time and waiting times can therefore be viewed in discrete steps, one
arrival per step. Multiplying gain per step by the total arrival rate
gives gain per unit time. Independent yes-or-no trials for each type
in a common slot can instead create simultaneous arrivals, requiring a
rule for their order or joint processing.

\subsection{Our contributions}

Our main result gives efficient static common prices for recurring XOS markets
with costly procurement and finite storage, even one slot per item.
Raising a common price changes seller participation and buyers' bundle choices,
hence the inventory available to future buyers.
A policy is \(\alpha\)-competitive if its expected long-run gain rate is at least
an \(\alpha\) fraction of the offline benchmark's.

Let \(C_{\min}\) be the smallest item capacity. With direct control of
allocations, sampling an optimal flow guarantees \(C_{\min}/(C_{\min}+1)\) of the
benchmark; our common prices guarantee one half while buyers choose their
preferred available bundles (Section~\ref{sec:lp-sampling}).

For computation, we use exact rational data and two polynomial-time procedures,
called oracles. A
\emph{demand oracle} returns a utility-maximizing bundle at given item prices.
An \emph{XOS oracle} returns a \emph{supporting XOS clause} for a specified
bundle: nonnegative item values that sum to the buyer's value for that bundle
and whose sum on any other bundle does not exceed the buyer's value for it.

\begin{theorem}[Informal multi-item result]
For XOS buyers and arbitrary finite item capacities, a static
contribution-balanced price vector is \(1/2\)-competitive against the
unlimited-inventory offline benchmark. It is computable in polynomial time
using demand and XOS oracles, and the guarantee holds under arbitrary
nonanticipatory tie-breaking.
\end{theorem}

We solve the flow relaxation and use supporting XOS clauses to split its
bundle values among items.
For item \(j\), let \(x_j\) be its total flow rate, \(W_j\) its total
assigned buyer value per unit time, and \(A_j\) its total seller cost per
unit time. We post the cost-adjusted midpoint
\[
p_j=\frac{W_j+A_j}{2x_j}
\]
when \(x_j>0\). The supporting-clause buyer comparison of Feldman, Gravin,
and Lucier~\cite{FGL15} gives the itemwise lower bound \(W_j-x_jp_j\).
To account for procurement, we add the seller-surplus bound \(x_jp_j-A_j\).
The midpoint makes both bounds equal to \((W_j-A_j)/2\).
Nonnegative utility and surplus let us replace a negative bound by zero.
Every item has stock or free capacity, so the combined bounds guarantee half
its flow contribution without calculating inventory probabilities.

The single-item analysis supplies a foundation and a sharp capacity comparison
for these pricing rules.

\begin{theorem}[Informal single-item result]
For every capacity \(C\ge1\), there is a rate-balanced posted price, together
with prescribed randomization when sellers or buyers are indifferent, that is
\(C/(C+1)\)-competitive against the unlimited-inventory offline benchmark.
This factor is best possible for every \(C\).
\end{theorem}

At a positive common intensity, all inventory levels are equally likely in
the long run; PASTA gives stock and free-capacity probabilities \(C/(C+1)\).
A one-seller, one-buyer family proves tightness.

Rate balance also extends to the multi-item model when the prices and
randomization for indifferent agents are chosen together. We prove that
buyers can find stock and sellers can find free capacity with equal
probabilities of at least one half for every item used by a suitably chosen
optimal flow solution (Theorem~\ref{thm:multi-rate-balance}).
In steady state, this also equalizes each such item's purchase rate conditional
on free capacity and its sale rate conditional on stock being available.
The proof uses
a fixed-point argument to choose prices consistent with the inventory
probabilities they generate. It handles the coupled effect of each price on buyers' bundle choices
and establishes existence rather than an efficient construction.
The equal probabilities of empty and full inventory are a possible starting
point for guarantees that improve with capacity; for capacities above one,
they need not make all inventory levels equally likely.

\subsection{Related work}
\label{sec:related}

\paragraph*{Stationary prophets.}
Kessel, Shameli, Saberi, and Wajc~\cite{KSSW22} formulate the stationary
prophet inequality, in which copies arrive automatically, may perish, and
are allocated to recurring buyers. Their Lemma~4.6 identifies the same
\(C/(C+1)\) capacity barrier: in their model, no capacity-\(C\) online policy
can guarantee a larger ratio relative to the optimal unbounded-inventory online
policy. Thus the numerical capacity phenomenon is not unique to our setting.
Our single-item theorem shows that it reappears in a two-sided welfare model
where inventory is built by costly purchases and does not perish, and that
a rate-balanced common posted price attains the factor against the
unlimited-inventory offline benchmark.

Patel and Wajc~\cite{PW24} extend the stationary-prophet line to
combinatorially constrained buyers and submodular objectives, and explicitly
identify XOS valuations and posted-price mechanisms as further directions. Our
contribution-balanced theorem treats a two-sided model with nonperishing stock
and XOS buyers: every stored unit is purchased from a recurring seller at the
posted price, the seller's cost enters social welfare, and buyers choose
bundles at the same item prices. These purchases and sales create the inventory
process.

\paragraph*{Trading models.}
Correa et al.~\cite{CCDHOS23} introduced the finite-horizon trading-prophet
model. Most closely, Aggarwal et al.~\cite{AGWZ26} separate buying and selling
opportunities and show that zero initial inventory precludes any finite ratio
in their asymmetric finite-horizon model, even for identical distributions.
Azar, Buchbinder, Levin, and Vardi~\cite{ABLV26Bundle} study a finite
adversarial bundle-trading model with online suppliers and customers, limited
inventories, demand oracles, a linear-programming benchmark, and inventory-dependent prices.
They obtain logarithmic profit guarantees under large-inventory and
value-augmentation assumptions. In contrast, our recurrent Poisson model has
single-item sellers, XOS buyers, and a long-run social-gain objective. Our
guarantees use static common prices and hold from empty initial inventory.
With finite storage, the terminal-inventory term vanishes in the long-run
average. The guarantees apply to arbitrary finite capacities.

\paragraph*{XOS posted prices.}
Feldman, Gravin, and Lucier~\cite{FGL15} derive static item prices for one-sided
combinatorial auctions using supporting additive clauses, and Braun and
Kesselheim~\cite{BK23SOSA} give a simplified treatment of this argument. We
inherit their supporting-clause decomposition and buyer-utility comparison.
Our setting adds recurring procurement: the same price controls seller
participation and buyers' demand, generating the inventory on which future
allocations depend. Our addition to the buyer comparison is a seller-surplus
certificate accounting for acquisition costs. Balancing the two bounds gives
the cost-adjusted midpoint.

D\"utting et al.~\cite{DKLRS24} establish \emph{median item prices}, which
balance the probabilities that each item is sold or left unsold in a
one-sided combinatorial auction. Their existence argument uses a fixed
point and suitable randomized choices for indifferent buyers. Their
polynomial-time construction is restricted to unit-demand buyers, who choose
at most one item; extending it to XOS buyers is left open in their paper. Our
rate-balance theorem controls stock and free-capacity probabilities in a
recurring inventory process. Here the prices affect both procurement from
sellers and buyers' bundle choices, and the prices and tie-breaking rule
are chosen together.

Braun and Kesselheim~\cite{BK23} also study strategic two-sided markets; our
agents are price-taking: they choose trades at the posted prices.
Strategic reporting and intermediary profit are outside the present model.

Appendix~\ref{app:extended-related} provides a broader discussion of classical
and stationary prophet inequalities, trading prophets, and two-sided
intermediation.

\subsection{Organization}

The remainder of the paper follows the two balancing principles. We first
analyze rate-balanced prices for one item, then both pricing principles for
XOS markets. We next compare the principles and then discuss limitations.
The flow relaxations and allocation baseline appear with the main results.
Technical sampling, fixed-point, and exact oracle-model details, together
with extended related work, are collected in the appendices.

\section{Rate-balanced prices for one item}
\label{sec:single-model}

\subsection{Model and unlimited-inventory offline benchmark}
\label{sec:single-arrivals}

There is one good. Seller type \(i\in I\) has cost \(a_i\ge0\) and arrives
according to a Poisson process of rate \(\lambda_i\). Buyer type \(j\in J\)
has value \(b_j\ge0\) and arrives according to an independent Poisson process
of rate \(\gamma_j\). The type sets are finite.

The intermediary has inventory capacity \(C\ge1\). Its inventory
\(N(t)\in\{0,1,\ldots,C\}\) increases by one after an accepted seller and
decreases by one after an accepted buyer. A purchase is possible only when
\(N(t)<C\), and a sale only when \(N(t)>0\). For a policy \(\mathcal A\), define
the expected long-run gain rate
\[
G(\mathcal A)
=
\liminf_{T\to\infty}\frac1T
\mathbb E\left[
\sum_{\text{sales before }T}b_j-
\sum_{\text{purchases before }T}a_i
\right].
\]
Here the objective subtracts seller costs, not cash transfers. A posted price
paid to a seller and later collected from a buyer cancels from long-run social
gain when purchase and sale flows balance.
The principal comparator, \(\OPT_{\mathrm{off}}^\infty\), is the
unlimited-inventory offline benchmark. It is clairvoyant and has unlimited
storage, but it must acquire each unit before allocating it. Both it and the
online policy start empty. As usual for a long-run average objective, a finite
initial transient is immaterial.

When the total arrival rate
\(\nu:=\sum_i\lambda_i+\sum_j\gamma_j\) is positive,
successive arrival types are independent, with seller type \(i\) occurring
with probability \(\lambda_i/\nu\) and buyer type \(j\) with probability
\(\gamma_j/\nu\). The waiting times between arrivals are memoryless, with
mean \(1/\nu\), and independent of these types and of each
other. For the policies using arrival types and inventory described in the
introduction, long-run gain per discrete step is converted to gain per unit
time by multiplying by \(\nu\).

The single-item relaxation matches seller and buyer types subject to their
supply and demand rates. Let \(x_{ij}\) be the relaxed trade rate between
seller type \(i\) and buyer type \(j\), measured in units per unit time.
These rates satisfy
\[
\sum_jx_{ij}\le\lambda_i,\qquad
\sum_ix_{ij}\le\gamma_j,\qquad
x_{ij}\ge0.
\]
The first constraint bounds each seller type's total trade rate by
\(\lambda_i\); the second bounds each buyer type's by \(\gamma_j\).
For a real number \(r\), its \emph{positive part} is
\(r^+:=\max\{r,0\}\). Consequently,
\begin{equation}
\OPT_{\mathrm{off}}^\infty\le \LP_1
:=
\max_x\left\{
\sum_{i,j}(b_j-a_i)^+x_{ij}:
\sum_jx_{ij}\le\lambda_i,\ 
\sum_ix_{ij}\le\gamma_j,\ 
x\ge0
\right\}.
\label{eq:single-lp}
\end{equation}
Here \(\LP_1\) is the maximum gain per unit of time in this flow relaxation.
Appendix~\ref{app:flow-upper-bounds} proves that it upper-bounds the
unlimited-inventory offline benchmark.

\paragraph*{Queueing theory background.}
The states of our inventory process are its possible inventory levels, from
empty to full. We say that the inventory is in \emph{steady state} when
the probability of each inventory level stays constant over time.
Different levels may have different probabilities. The probability distribution
over inventory levels in steady state is called a \emph{stationary
distribution}. For a process with finitely many states, such a distribution
exists and is unique if the rates of moving between states depend only on
the current state and every state can be reached from every
other~\cite[Section~3.5]{Norris97}. Our single-item policy has finitely many
inventory levels and independent Poisson arrivals. It uses the same price
and acceptance rules at every arrival, with independent random choices.
Its rates of moving between
inventory levels therefore depend only on the current inventory. When the
purchase and sale rates are positive, every inventory level can be reached
from every other, so these conditions hold. The zero-rate case is handled
separately in the single-item proof.

The Poisson Arrivals See Time Averages property (PASTA)~\cite{Wolff82}
lets us translate these inventory probabilities into arrival probabilities.

\begin{lemma}[PASTA]
\label{lem:poisson-arrivals}
In steady state, the probability that Poisson arrivals see a process in a
given state immediately before their arrival equals the probability of that
state at an arbitrary time, provided the process does not anticipate future
arrivals.
\end{lemma}

Thus arriving buyers see the same stock-availability probability as at an
arbitrary time, and arriving sellers see the same probability of free
capacity. Inventory may depend on past arrivals; future arrivals must remain
independent of that history. A formal version and proof appear in
Appendix~\ref{app:finite-horizon-multi}, which also proves the multi-item
guarantee without assuming steady state.

\subsection{A rate-balanced posted price}
\label{sec:balanced}

Fix a posted price \(\tau\). The intermediary pays \(\tau\) to a trading seller
and charges \(\tau\) to a trading buyer. A seller trades when \(a_i<\tau\), a
buyer trades when \(b_j>\tau\), and types on the other side reject. At equality
the policy prescribes trade with probability \(p_\tau\in[0,1]\) for a seller
and \(q_\tau\in[0,1]\) for a buyer. Write the resulting acceptance
probabilities as
\[
\alpha_i=
\begin{cases}
1&a_i<\tau,\\
p_\tau&a_i=\tau,\\
0&a_i>\tau,
\end{cases}
\qquad
\beta_j=
\begin{cases}
1&b_j>\tau,\\
q_\tau&b_j=\tau,\\
0&b_j<\tau.
\end{cases}
\]
Define
\[
u=\sum_i\alpha_i\lambda_i,
\qquad
v=\sum_j\beta_j\gamma_j.
\]
We call \(u\) and \(v\) the potential purchase and sale intensities. Whenever
the inventory is not full, upward transitions occur at rate \(u\); whenever it
is nonempty, downward transitions occur at rate \(v\). These are not realized
throughput rates: over a horizon of length \(T\), the normalized numbers of
purchases and sales differ by at most \(C/T\) because inventory is bounded. We
call the resulting posted-price policy rate-balanced when \(u=v\).

\begin{lemma}[Rate-balanced posted price]
\label{lem:balanced-exists}
There are \(\tau,p_\tau,q_\tau\) for which \(u=v=:R\).
\end{lemma}

\begin{proof}[Proof sketch]
For a price \(z\), let
\[
L_<(z)=\sum_{i:a_i<z}\lambda_i,\qquad
L_{\le}(z)=\sum_{i:a_i\le z}\lambda_i,
\]
and
\[
H_>(z)=\sum_{j:b_j>z}\gamma_j,\qquad
H_{\ge}(z)=\sum_{j:b_j\ge z}\gamma_j.
\]
At \(z\), seller tie-breaking generates every rate in
\([L_<(z),L_{\le}(z)]\), and buyer tie-breaking generates every rate in
\([H_>(z),H_{\ge}(z)]\). Choose \(\tau\) as the first price level at which
\(L_{\le}(\tau)\ge H_>(\tau)\). Immediately below that level,
\(L_<(\tau)\le H_{\ge}(\tau)\). The two intervals therefore intersect.
Choose \(R\) in their intersection and set the two tie probabilities to
realize \(R\). If no type on one side has cost or value \(\tau\), that interval is a
singleton and its tie probability is irrelevant.
\end{proof}

Fix a rate-balanced price and tie probabilities from
Lemma~\ref{lem:balanced-exists}. For capacity \(C\), let
\(\mathcal A_{\tau,C}\) denote the resulting posted-price policy, and write
\(G_C:=G(\mathcal A_{\tau,C})\).

\begin{theorem}[Tight single-item guarantee]
\label{thm:single-main}
For every \(C\ge1\), the rate-balanced posted-price policy satisfies
\[
G_C=\frac{C}{C+1}\LP_1
\ge\frac{C}{C+1}\OPT_{\mathrm{off}}^\infty.
\]
Moreover, no online policy with capacity \(C\) can guarantee a larger ratio
against \(\OPT_{\mathrm{off}}^\infty\).
\end{theorem}

\begin{proof}
We first prove the guarantee. If \(R=0\), the absence of accepted strict types
implies
\[
\sum_j\gamma_j(b_j-\tau)^+
+\sum_i\lambda_i(\tau-a_i)^+=0.
\]
The positive-part bound below and feasibility of the zero matching then
give \(G_C=\LP_1=0\).
Assume henceforth that \(R>0\).

The inventory is a birth-death chain: a purchase adds one unit and a sale
removes one unit, with each feasible transition having rate \(R\).
Since \(R>0\), it is irreducible (every inventory level can be reached
from every other). Let \(\pi_n\) be the steady-state probability of inventory
level \(n\). Flow between neighboring levels balances:
\(\pi_nR=\pi_{n+1}R\) for \(n=0,\ldots,C-1\).
Thus all \(C+1\) levels have the same probability, giving the unique
stationary distribution~\cite[Sections~3.5 and~3.7]{Norris97}:
\[
\pi_n=\frac1{C+1},\qquad n=0,\ldots,C.
\]
By PASTA (Lemma~\ref{lem:poisson-arrivals}), arriving buyers find stock
with probability \(1-\pi_0\), and arriving sellers find free capacity with
probability \(1-\pi_C\). Both probabilities equal \(C/(C+1)\).
With independent acceptance coins at equality,
type-\(j\) buyers trade at rate
\(\beta_j\gamma_j(1-\pi_0)\), and type-\(i\) sellers trade at rate
\(\alpha_i\lambda_i(1-\pi_C)\). Summing buyer values minus seller costs
therefore gives
\begin{align}
G_C
&=\frac{C}{C+1}
\left(\sum_j\beta_j\gamma_jb_j-
      \sum_i\alpha_i\lambda_ia_i\right)\notag\\
&=\frac{C}{C+1}
\left(\sum_j\gamma_j(b_j-\tau)^+
      +\sum_i\lambda_i(\tau-a_i)^+\right),
\label{eq:single-gain}
\end{align}
For the second equality, add and subtract \(R\tau\), using
\(\sum_j\beta_j\gamma_j=\sum_i\alpha_i\lambda_i=R\).
The remaining terms satisfy
\(\beta_j(b_j-\tau)=(b_j-\tau)^+\) and
\(\alpha_i(\tau-a_i)=(\tau-a_i)^+\): strict acceptable types trade with
probability one, rejected types contribute zero, and equality types have
zero surplus regardless of their tie probabilities.

For every seller-buyer pair,
\[
(b_j-a_i)^+\le(b_j-\tau)^++(\tau-a_i)^+.
\]
Thus every feasible \(x\) in~\eqref{eq:single-lp} satisfies
\begin{align*}
\sum_{i,j}(b_j-a_i)^+x_{ij}
&\le
\sum_j(b_j-\tau)^+\sum_ix_{ij}
+\sum_i(\tau-a_i)^+\sum_jx_{ij}\\
&\le
\sum_j\gamma_j(b_j-\tau)^+
+\sum_i\lambda_i(\tau-a_i)^+.
\end{align*}
The first inequality applies the preceding pairwise bound and rearranges the
sums. The second uses the LP constraints
\(\sum_i x_{ij}\le\gamma_j\) and
\(\sum_j x_{ij}\le\lambda_i\), together with the nonnegativity of the
positive-part coefficients.
A feasible flow matching the accepted rates attains this upper bound.
Define
\[
x_{ij}:=\frac{(\alpha_i\lambda_i)(\beta_j\gamma_j)}{R}.
\]
Its marginal sums satisfy
\[
\sum_jx_{ij}=\alpha_i\lambda_i\le\lambda_i,
\qquad
\sum_ix_{ij}=\beta_j\gamma_j\le\gamma_j,
\]
so it is feasible for~\eqref{eq:single-lp}. Whenever \(x_{ij}>0\), the
acceptance rules imply \(a_i\le\tau\le b_j\). Therefore
\begin{align*}
\sum_{i,j}(b_j-a_i)^+x_{ij}
&=\sum_j\beta_j\gamma_jb_j-\sum_i\alpha_i\lambda_ia_i\\
&=\sum_j\gamma_j(b_j-\tau)^+
  +\sum_i\lambda_i(\tau-a_i)^+.
\end{align*}
The second equality again adds and subtracts \(R\tau\). Thus the matching
attains the LP upper bound, and~\eqref{eq:single-gain} proves the claimed
equality.

To establish tightness in our two-sided model, fix \(\varepsilon>0\) and consider
one free seller type of rate \(1+\varepsilon\) and one buyer type of value \(1\)
and rate \(1\). Since purchases are free and all sales are identical, accepting
every feasible arrival maximizes the gain of any online capacity-\(C\) policy.
Its inventory is a birth-death chain with rate ratio \(\rho=1+\varepsilon\),
whose empty-state probability is
\[
\pi_0=\frac{\rho-1}{\rho^{C+1}-1}
=\frac{\varepsilon}{(1+\varepsilon)^{C+1}-1}.
\]
The optimal online gain is \(1-\pi_0\). The unlimited-inventory offline
benchmark stores the free oversupply and serves asymptotically every buyer, so
its gain rate is \(1\). Sending \(\varepsilon\downarrow0\) gives
\[
1-\pi_0\longrightarrow1-\frac1{C+1}=\frac{C}{C+1}.\qedhere
\]
\end{proof}

At \(C=1\), the uniform stationary distribution gives both buyers and sellers
access probability \(1/2\), providing the median-style interpretation of rate
balance.

\section{Contribution-balanced prices for XOS markets}
\label{sec:multi-model}

\subsection{Model, policy, and flow benchmark}
\label{sec:multi-benchmark}

Let \(M=\{1,\ldots,m\}\). Item \(j\) has an arbitrary finite inventory
capacity \(C_j\ge1\). Writing \(N(t^-)\) for inventory immediately before
time \(t\), we have
\[
N(t)\in\prod_{j\in M}\{0,\ldots,C_j\},
\qquad
A_t=\{j:N_j(t^-)>0\},
\]
where \(A_t\) is the set of available items. Seller type
\(r\in R_j\) has cost \(a_{jr}\ge0\), can supply one unit of item \(j\), and
arrives at rate \(\lambda_{jr}\). We assume positive supply for every item:
\[
\Lambda_j:=\sum_{r\in R_j}\lambda_{jr}>0.
\]

Buyer type \(k\in K\) arrives at rate \(\gamma_k\). Its valuation
\(v_k:2^M\to\R_+\) is XOS: there is a finite family
of nonnegative additive clauses \(\mathcal A_k\) such that
\[
v_k(S)=\max_{w\in\mathcal A_k}\sum_{j\in S}w_j.
\]
The empty bundle has value zero, and adding items cannot decrease value.
Each buyer consumes at most one unit of every item in its selected bundle.
All type sets are finite, the Poisson processes are independent, and
\(\Gamma:=\sum_k\gamma_k>0\). Social gain is buyer value minus seller cost;
cash transfers do not enter this objective.

The discrete-time interpretation in the introduction applies to these
streams as well for rules that do not use timing information.

At a price vector \(p\in\R^m\), an arriving seller of item \(j\) is offered
\(p_j\). When \(N_j(t^-)<C_j\), the seller trades and receives \(p_j\) if
\(a_{jr}<p_j\), rejects if \(a_{jr}>p_j\), and may take either action at
equality according to any nonanticipatory rule. An arriving buyer chooses a
utility-maximizing available bundle
\[
D_k(A_t,p)\in\arg\max_{S\subseteq A_t}
\left\{v_k(S)-\sum_{j\in S}p_j\right\},
\]
and pays \(\sum_{j\in D_k(A_t,p)}p_j\). Tie-breaking among maximizers may be
randomized and may depend on calendar time, the current state, and the entire
past, provided it is nonanticipatory. Allocating bundle \(D_k(A_t,p)\) removes
one unit of every item in that bundle. Thus \(p_j\) is the common price paid to
a seller and charged to a buyer of item \(j\).

The multi-item relaxation again ignores arrival order and inventory
limits. Its variable \(y_{kS}\) is the bundle-\(S\) allocation rate for buyer
type \(k\), in bundles per unit time; \(z_{jr}\) is the purchase rate from
seller type \(r\) of item \(j\), in units of that item per unit time.
The constraints are
\begin{align}
\sum_{S\subseteq M}y_{kS}&\le\gamma_k &&\forall k,
\label{eq:multi-buyer}\\
\sum_k\sum_{S\ni j}y_{kS}&=\sum_{r\in R_j}z_{jr} &&\forall j,
\label{eq:multi-flow}\\
y_{kS}&\ge0 &&\forall k,\ S\subseteq M,
\label{eq:multi-nonnegative}\\
0\le z_{jr}&\le\lambda_{jr} &&\forall j,r.
\label{eq:multi-seller}
\end{align}
Constraint~\eqref{eq:multi-buyer} limits bundle allocation rates by buyer
arrival rates. Constraint~\eqref{eq:multi-flow} matches each item's purchase
and allocation rates, counting one unit per bundle containing it. The
remaining constraints make rates nonnegative and bound purchases by seller
arrival rates.
The maximum gain per unit of time in this relaxation is
\begin{equation}
\LP_m=
\max\left\{
\sum_{k,S}y_{kS}v_k(S)-\sum_{j,r}z_{jr}a_{jr}:
\eqref{eq:multi-buyer}\text{ through }\eqref{eq:multi-seller}
\right\}.
\label{eq:multi-lp}
\end{equation}

For any online or offline policy \(\mathcal A\), let
\(\mathrm{gain}_{\mathcal A}(0,T)\) be its total gain over \([0,T]\), and define
\[
G(\mathcal A):=\liminf_{T\to\infty}\frac1T
\mathbb E[\mathrm{gain}_{\mathcal A}(0,T)].
\]
For a price vector \(p\) and an admissible tie-breaking rule fixed in context,
let \(\mathcal A_p\) be the induced posted-price policy, and abbreviate
\[
\mathrm{gain}_p(0,T):=\mathrm{gain}_{\mathcal A_p}(0,T),
\qquad
G(p):=G(\mathcal A_p).
\]
Let \(\OPT_{\mathrm{off},m}^{\infty}\) be the supremum of \(G(\mathcal A)\)
over clairvoyant offline policies with unlimited inventory that start empty
and allocate only previously acquired units.
Let \(\OPT_{\mathrm{on},m}^{\infty}\) be the supremum of \(G(\mathcal A)\)
over randomized nonanticipatory online policies with the same storage and
temporal constraints. They observe arriving types and may directly choose
any available bundle for each buyer.

Appendix~\ref{app:flow-upper-bounds} proves that every temporally feasible
online or clairvoyant policy starting empty satisfies, for every \(T>0\),
\[
\mathbb E[\mathrm{gain}_{\mathcal A}(0,T)]\le T \cdot \LP_m.
\]
Sampling an optimal flow gives an
unrestricted online allocation policy with long-run gain at least
\(C_{\min}\LP_m/(C_{\min}+1)\)
(Proposition~\ref{prop:lp-sampling}, Section~\ref{sec:lp-sampling}).
It starts empty, respects the capacities, and directly chooses buyers'
allocations. With unlimited storage, it may impose arbitrarily large finite
common capacities, giving the benchmark equality
\(\OPT_{\mathrm{on},m}^{\infty}=\OPT_{\mathrm{off},m}^{\infty}=\LP_m\)
(Corollary~\ref{cor:unlimited-lp}).
The same LP upper-bounds every finite-capacity offline policy.

\subsection{The utility-surplus decomposition}
\label{sec:decomposition}

We explain the argument in steady state, writing \(N\) for the inventory
vector. Appendix~\ref{app:finite-horizon-multi} gives the full proof for
arbitrary nonanticipatory tie-breaking.

Define
\[
q_j^B:=\Pr[N_j>0],
\qquad
q_j^S:=\Pr[N_j<C_j].
\]
Thus \(q_j^B\) is the probability that item \(j\) is available to a buyer,
while \(q_j^S\) is the probability that it has free capacity for a seller.
We make no independence assumption on the coordinates
of \(N\). Since every inventory state satisfies
\[
\one\{N_j>0\}+\one\{N_j<C_j\}\ge1,
\]
we have
\begin{equation}
q_j^B+q_j^S\ge1.
\label{eq:access-sum}
\end{equation}

Fix a feasible flow solution \((y,z)\) in~\eqref{eq:multi-lp}, and let
\[
\mathcal S(y):=\{(k,S):y_{kS}>0\}
\]
be the support of its bundle variables. For every
\((k,S)\in\mathcal S(y)\), fix a supporting XOS clause \(w^{kS}\):
\[
v_k(S)=\sum_{j\in S}w^{kS}_j,
\qquad
v_k(T)\ge\sum_{j\in T}w^{kS}_j
\quad\text{for every }T\subseteq M.
\]
For each item define the buyer-utility and seller-surplus potentials
\begin{align}
B_j(p_j)
&:=\sum_{\substack{(k,S)\in\mathcal S(y)\\j\in S}}
y_{kS}(w^{kS}_j-p_j)^+,
\label{eq:buyer-potential}\\
L_j(p_j)
&:=\sum_{r\in R_j}\lambda_{jr}(p_j-a_{jr})^+.
\label{eq:seller-potential}
\end{align}

\begin{lemma}[Long-run utility-surplus decomposition]
\label{lem:stationary-decomposition}
In steady state, every nonnegative price vector \(p\) satisfies
\begin{equation}
G(p)\ge\sum_{j\in M}
\left(q_j^B B_j(p_j)+q_j^S L_j(p_j)\right).
\label{eq:stationary-decomposition}
\end{equation}
\end{lemma}

\begin{proof}
The buyer-side argument follows Feldman, Gravin, and Lucier~\cite{FGL15}.
For an arrival of type \(k\), independently draw an analysis-only reference
bundle \(S\) with probability \(y_{kS}/\gamma_k\), using the empty bundle for
the remaining probability. Types with \(\gamma_k=0\) have no positive flow
and are ignored. Given the available set \(A\), the buyer could choose
\[
T_{kS}(A):=\{j\in S\cap A:w^{kS}_j\ge p_j\}.
\]
Optimality of the demanded bundle and the supporting-clause inequality give
\begin{align*}
v_k(D_k(A,p))-\sum_{j\in D_k(A,p)}p_j
&\ge v_k(T_{kS}(A))-\sum_{j\in T_{kS}(A)}p_j\\
&\ge\sum_{j\in S}(w^{kS}_j-p_j)^+\one\{j\in A\}.
\end{align*}

Let \(U_B(p)\) denote the long-run buyer-utility rate. By PASTA
(Lemma~\ref{lem:poisson-arrivals}), type-\(k\) buyers find item \(j\)
available with probability \(q_j^B\). The analysis-only reference bundle is
drawn independently of inventory and arrivals. Averaging the preceding
inequality over these bundles and buyer types therefore gives
\[
U_B(p)\ge\sum_jq_j^B B_j(p_j).
\]

Similarly, let \(U_S(p)\) denote the long-run seller-surplus rate. A seller
with \(a_{jr}<p_j\) trades whenever item \(j\) is nonfull, while a seller at
equality contributes zero surplus regardless of its action. By PASTA,
type-\((j,r)\) sellers find free capacity with probability \(q_j^S\),
so they arrive with free capacity at rate \(\lambda_{jr}q_j^S\).
Multiplying by the surplus
\((p_j-a_{jr})^+\) and summing yields
\[
U_S(p)=\sum_jq_j^S L_j(p_j).
\]
Neither application requires independence between the item inventories.

Because inventory is bounded, purchases and allocations of every item have
the same long-run rate. Payments at the common posted prices therefore cancel,
and
\[
G(p)=U_B(p)+U_S(p).
\]
Combining the three displays proves~\eqref{eq:stationary-decomposition}. The
finite-horizon accounting identity behind the last equality is recorded in
Appendix~\ref{app:finite-horizon-multi}.
\end{proof}

\subsection{Two routes to a one-half guarantee}
\label{sec:constructive-prices}

For the fixed feasible solution and supporting clauses above, define
\begin{align*}
x_j
&:=\sum_k\sum_{S\ni j}y_{kS}
=\sum_{r\in R_j}z_{jr},\\
W_j
&:=\sum_{\substack{(k,S)\in\mathcal S(y)\\j\in S}}
y_{kS}w^{kS}_j,\\
A_j
&:=\sum_{r\in R_j}z_{jr}a_{jr}.
\end{align*}
Here \(x_j\) is the LP flow of item \(j\); when \(x_j>0\), the ratios
\(W_j/x_j\) and \(A_j/x_j\) are its average supported buyer value and average
seller cost.

The supporting clauses decompose the value of the flow solution itemwise:
\begin{align}
\sum_{k,S}y_{kS}v_k(S)-\sum_{j,r}z_{jr}a_{jr}
&=\sum_j(W_j-A_j).
\label{eq:itemwise-lp-master}
\end{align}

\begin{lemma}[Master itemwise bound]
\label{lem:master-itemwise}
In steady state, every nonnegative price vector \(p\) satisfies
\begin{equation}
G(p)\ge
\sum_j\left[
q_j^B(W_j-x_jp_j)+q_j^S(x_jp_j-A_j)
\right].
\label{eq:master-itemwise}
\end{equation}
\end{lemma}

\begin{proof}
Dropping positive parts in the buyer potential gives
\[
B_j(p_j)\ge W_j-x_jp_j.
\]
Moreover, \(0\le z_{jr}\le\lambda_{jr}\) and
\((p_j-a_{jr})^+\ge p_j-a_{jr}\), so
\[
L_j(p_j)
\ge\sum_r z_{jr}(p_j-a_{jr})^+
\ge x_jp_j-A_j.
\]
Substitution into Lemma~\ref{lem:stationary-decomposition} proves the claim.
\end{proof}

The first route extends rate balance to several items by controlling
inventory availability.

\begin{corollary}[Rate-balance route]
\label{cor:half-availability}
In steady state, suppose a nonnegative price vector satisfies
\[
\min\{q_j^B,q_j^S\}\ge\alpha
\]
for every item with positive LP flow, where \(\alpha\in[0,1]\). Then
\[
G(p)\ge\alpha
\left(\sum_{k,S}y_{kS}v_k(S)-\sum_{j,r}z_{jr}a_{jr}\right).
\]
In particular, \(\alpha=1/2\) and an optimal flow solution yield a \(1/2\)
guarantee.
\end{corollary}

\begin{proof}
Since the potentials are nonnegative,
Lemma~\ref{lem:stationary-decomposition} gives
\[
G(p)\ge\alpha
\sum_{j:x_j>0}\bigl(B_j(p_j)+L_j(p_j)\bigr).
\]
The two lower bounds in the proof of Lemma~\ref{lem:master-itemwise} imply
\[
B_j(p_j)+L_j(p_j)\ge W_j-A_j.
\]
Summing and using~\eqref{eq:itemwise-lp-master} proves the claim.
\end{proof}

At unit capacity, \(q_j^S=1-q_j^B\), so the \(\alpha=1/2\) condition in
Corollary~\ref{cor:half-availability} requires half availability.
For one item,
Lemma~\ref{lem:balanced-exists} constructs the required price and tie-breaking.
The following theorem supplies simultaneous rate balance for several items
when the randomized tie-breaking rule is chosen together with the prices.

\begin{theorem}[Existence of multi-item rate-balanced prices]
\label{thm:multi-rate-balance}
For XOS buyers and arbitrary finite item capacities, there exist nonnegative
item prices, a fixed randomized tie-breaking rule, and an optimal flow
solution such that the inventory process starting empty converges to a
stationary distribution satisfying
\[
\Pr[N_j=0]=\Pr[N_j=C_j],
\qquad q_j^B=q_j^S\ge\frac12
\]
for every item with positive flow in that solution. Seller tie probabilities
depend only on seller type; buyer tie probabilities depend only on buyer
type and the available items. The resulting policy satisfies
\[
G(p)\ge\frac12\LP_m
\ge\frac12\OPT_{\mathrm{off},m}^{\infty}.
\]
\end{theorem}

Appendix~\ref{app:multi-rate-balance} gives the full fixed-point proof
of Theorem~\ref{thm:multi-rate-balance}. It establishes existence; a
polynomial-time construction remains open.

The second route balances the two certified contributions rather than the
purchase and sale rates.
For \(x_j>0\), set
\begin{equation}
p_j:=\frac{W_j+A_j}{2x_j};
\label{eq:constructive-price}
\end{equation}
if \(x_j=0\), set \(p_j=0\).

For every positive-flow item, this price satisfies
\begin{equation}
W_j-x_jp_j=x_jp_j-A_j=\frac{W_j-A_j}{2}.
\label{eq:contribution-balance}
\end{equation}

\begin{theorem}[Contribution-balanced multi-item guarantee]
\label{thm:constructive-multi}
For arbitrary finite capacities \(C_j\ge1\), the posted-price policy
using~\eqref{eq:constructive-price}, starting empty and under any admissible
nonanticipatory tie-breaking rule, satisfies
\[
G(p)\ge\frac12
\left(\sum_{k,S}y_{kS}v_k(S)-\sum_{j,r}z_{jr}a_{jr}\right).
\]
For an optimal flow solution, every horizon \(T>0\) also satisfies
\[
\mathbb E[\mathrm{gain}_p(0,T)]
\ge\frac{T}{2}\LP_m-\sum_jp_jC_j.
\]
The subtracted term bounds the posted-price value of inventory purchased but
not yet sold. Hence \(G(p)\ge\LP_m/2\ge\OPT_{\mathrm{off},m}^{\infty}/2\).
These bounds do not require a stationary inventory distribution.
\end{theorem}

\begin{proof}
Fix \(T>0\). Let \(q_j^B(T)\) and \(q_j^S(T)\) be the expected fractions
of time in \([0,T]\) during which item \(j\) has stock and free capacity,
respectively. Thus \(q_j^B(T)+q_j^S(T)\ge1\).

Let \(d_j=(W_j-A_j)/2\); zero-flow items have \(W_j=A_j=d_j=0\).
For positive-flow items, \eqref{eq:contribution-balance} and the bounds in
Lemma~\ref{lem:master-itemwise}'s proof give
\(B_j(p_j)\ge d_j\) and \(L_j(p_j)\ge d_j\).
Nonnegativity of the potentials then gives, for every item,
\[
q_j^B(T)B_j(p_j)+q_j^S(T)L_j(p_j)
\ge\bigl(q_j^B(T)+q_j^S(T)\bigr)d_j^+
\ge d_j.
\]
Let \(\mathcal U_B(T)\) and \(\mathcal U_S(T)\) denote total buyer
utility and seller surplus over \([0,T]\). Since inventory starts empty,
the common-price payments give the pathwise identity
\[
\mathrm{gain}_p(0,T)
=\mathcal U_B(T)+\mathcal U_S(T)-\sum_jp_jN_j(T)
\ge\mathcal U_B(T)+\mathcal U_S(T)-\sum_jp_jC_j.
\]
The finite-horizon Poisson sampling calculation in
Appendix~\ref{app:finite-horizon-multi} gives
\[
\mathbb E[\mathcal U_B(T)+\mathcal U_S(T)]
\ge T\sum_j
\left(q_j^B(T)B_j(p_j)+q_j^S(T)L_j(p_j)\right)
\ge T\sum_jd_j.
\]
By~\eqref{eq:itemwise-lp-master}, \(\sum_jd_j\) is half the flow
objective. This proves the finite-horizon claim. Dividing by \(T\),
taking the liminf, and using the flow upper bound proves the
long-run claims.
\end{proof}

\paragraph*{Uniform tightness.}
Against the unlimited-inventory offline benchmark, the factor \(1/2\) is
optimal as a guarantee uniform over all finite capacity vectors: the model
contains the one-item, unit-capacity case, for which
Theorem~\ref{thm:single-main} rules out any larger ratio even for arbitrary
online policies. This does not rule out stronger capacity-sensitive posted-price guarantees
when every capacity exceeds one.

\subsection{Allocation baseline and benchmark equality}
\label{sec:lp-sampling}

The following policy directly chooses the bundles allocated to buyers.

\begin{proposition}[Unrestricted LP-sampling baseline]
\label{prop:lp-sampling}
Let \(B:=\min_j C_j\). For every feasible solution \((y,z)\) of
\eqref{eq:multi-lp}, there is a randomized online allocation policy
\(\mathcal A\), starting empty and respecting the inventory capacities, such that
\[
G(\mathcal A)\ge
\frac{B}{B+1}
\left(
\sum_{k,S}y_{kS}v_k(S)-\sum_{j,r}z_{jr}a_{jr}
\right).
\]
In particular, an optimal flow solution gives
\(G(\mathcal A)\ge B\LP_m/(B+1)\).
\end{proposition}

\begin{proof}
Use only \(B\) storage slots for each item. Independently retain each
type-\((j,r)\) seller with probability \(z_{jr}/\lambda_{jr}\), and purchase
its unit if retained and \(N_j(t^-)<B\). Independently assign each
type-\(k\) buyer a reference bundle \(S\) with probability
\(y_{kS}/\gamma_k\), using the empty bundle for the remaining probability.
Allocate \(S\cap A_t\) to that buyer. All draws use fresh independent
randomness, without using pre-arrival inventory or previous history.
Zero-rate types have zero corresponding flow and are ignored.
The probabilities are valid by~\eqref{eq:multi-buyer} and
\eqref{eq:multi-seller}. Every allocated unit is present in inventory,
and purchases respect the effective capacity \(B\), so the policy is
temporally feasible and nonanticipatory.

Fix supporting XOS clauses for the positive bundle variables and use
\(x_j,W_j,A_j\) as defined in Section~\ref{sec:constructive-prices}.
For each item with \(x_j>0\), retained sellers form a Poisson stream of rate
\(\sum_r z_{jr}=x_j\), while buyers whose reference bundles contain
\(j\) form an independent Poisson stream of rate
\(\sum_{k,S\ni j}y_{kS}=x_j\).
A retained seller increases inventory whenever it is below \(B\);
a reference request decreases it whenever it is positive. Thus the
marginal inventory process of item \(j\) is a balanced birth-death
chain on \(\{0,\ldots,B\}\), with both feasible transition rates
equal to \(x_j\). Its unique stationary distribution is uniform,
and its distribution converges to that law from empty inventory.
Consequently, both \(\Pr[N_j(t^-)>0]\) and \(\Pr[N_j(t^-)<B]\)
converge to \(f:=B/(B+1)\).
Reference bundles may correlate different inventory coordinates;
the marginal argument does not require their independence.
If \(x_j=0\), the item is never purchased or requested, and
\(W_j=A_j=0\).

Let \(\mathcal V(T)\) and \(\mathcal C(T)\) be the total buyer value and
seller cost generated by the policy through time \(T\).
For a buyer assigned \(S\), its supporting clause gives
\[
v_k(S\cap A_t)\ge
\sum_{j\in S}w_j^{kS}\one\{N_j(t^-)>0\}.
\]
Averaging over the independent draws and applying
Lemma~\ref{lem:poisson-compensation} to the arrival streams yields
\begin{align*}
\mathbb E[\mathcal V(T)]&\ge
\sum_j W_j\int_0^T\Pr[N_j(t^-)>0]\,dt,\\
\mathbb E[\mathcal C(T)]&=
\sum_j A_j\int_0^T\Pr[N_j(t^-)<B]\,dt.
\end{align*}
Divide by \(T\). For every positive-flow item, the two averaged
probabilities converge to \(f\); zero-flow items contribute nothing.
Since \(\mathrm{gain}_{\mathcal A}(0,T)=\mathcal V(T)-\mathcal C(T)\),
taking the liminf and using~\eqref{eq:itemwise-lp-master} proves the claim.
\end{proof}

This policy controls the allocated subset directly. The guarantee does not
establish an implementation through common posted prices under which buyers
choose utility-maximizing available bundles.

\begin{corollary}[Unlimited-inventory benchmark]
\label{cor:unlimited-lp}
The unlimited-inventory online and offline benchmarks satisfy
\[
\OPT_{\mathrm{on},m}^{\infty}
=\OPT_{\mathrm{off},m}^{\infty}
=\LP_m.
\]
\end{corollary}

\begin{proof}
For any integer \(B\ge1\), apply Proposition~\ref{prop:lp-sampling}
to the same arrival rates, costs, and valuations, with capacity \(B\)
for every item. The relaxation~\eqref{eq:multi-lp} does not depend on
the capacities, so its value remains \(\LP_m\). An unlimited-inventory
online policy may voluntarily impose these finite capacities. Hence
\[
\OPT_{\mathrm{on},m}^{\infty}\ge\frac{B}{B+1}\LP_m
\qquad\text{for every }B\ge1.
\]
Letting \(B\) grow gives \(\OPT_{\mathrm{on},m}^{\infty}\ge\LP_m\).
Every online policy is also admissible for the offline benchmark, and the
flow upper bound in Appendix~\ref{app:flow-upper-bounds} gives
\(\OPT_{\mathrm{off},m}^{\infty}\le\LP_m\), proving the equality.
\end{proof}

\subsection{Polynomial-time computation}

The flow, supporting clauses, and prices can be computed efficiently under
the following exact oracle assumptions.

\begin{theorem}[Computability]
\label{thm:computability}
Assume rational rates and seller costs and exact demand and XOS oracles for
each buyer type. Both oracles run in polynomial time, with the uniform binary
encoding bounds on data and returned supporting clauses specified in
Appendix~\ref{app:oracle-details}. Let \(d:=|K|+m+\sum_j|R_j|\).
An exact optimal solution of~\eqref{eq:multi-lp} is computable in time
polynomial in \(d\) and these encoding bounds, with at most \(d\) positive
variables among \(y\) and \(z\). In particular, at most \(d\) bundle variables
\(y_{kS}\) are positive. The prices in~\eqref{eq:constructive-price} have the
same computation bound; each buyer's choice uses a demand query.
\end{theorem}

\begin{proof}[Proof sketch]
The dual of~\eqref{eq:multi-lp} has variables \(\alpha_k\ge0\),
\(\mu_j\ge0\), and \(\eta_{jr}\ge0\):
\begin{align}
\min\quad&
\sum_k\gamma_k\alpha_k+\sum_{j,r}\lambda_{jr}\eta_{jr}
\label{eq:multi-dual-objective}\\
\text{subject to}\quad&
\alpha_k+\sum_{j\in S}\mu_j\ge v_k(S)
&&\forall k,\ S\subseteq M,
\label{eq:multi-dual-buyer}\\
&\eta_{jr}\ge\mu_j-a_{jr}
&&\forall j,r.
\label{eq:multi-dual-seller}
\end{align}
The item-flow equalities initially give free variables \(\mu_j\), but they may
be replaced by their positive parts without increasing the objective or
violating any constraint because valuations, seller costs, and \(\eta\) are
nonnegative.

For a buyer type \(k\), a demand query at prices \(\mu\) returns a bundle
maximizing \(v_k(S)-\sum_{j\in S}\mu_j\). An XOS query evaluates the returned
bundle and supplies a supporting clause. This either identifies a violated
constraint~\eqref{eq:multi-dual-buyer} or certifies all buyer constraints for
that type. The explicit constraints are checked directly, giving a strong separation
oracle. Standard optimization--separation and primal-recovery results of
Gr\"otschel, Lov\'asz, and Schrijver~\cite[Theorem~6.4.9 and
Lemma~6.5.15]{GLS88} yield an optimal primal solution with polynomially many
nonzero variables. A harmless itemwise trimming of seller flow restores the equalities
in~\eqref{eq:multi-flow} after the restriction \(\mu\ge0\).
Appendix~\ref{app:oracle-details} gives the exact oracle-model,
primal-recovery, support, and encoding details.

For every positive primal variable \(y_{kS}\), query the XOS oracle on \(S\)
to obtain \(w^{kS}\). The sums \(x_j,W_j,A_j\) and the prices
in~\eqref{eq:constructive-price} are then computed directly. Demand queries
also implement each arriving buyer's choice under these prices.
\end{proof}

\begin{remark}[Approximate optimization]
If a flow is feasible for the true arrival rates and its objective under the
true valuations and costs is at least \(\LP_m-\varepsilon\), the price
construction using exact supporting clauses gives
\[
G(p)\ge\frac12(\LP_m-\varepsilon).
\]
Establishing these conditions from numerical approximations or samples, or
replacing the exact oracles, requires a separate argument.
\end{remark}

\section{Balancing principles: comparison and open directions}
\label{sec:comparison}

The unlimited-storage online and clairvoyant values both equal \(\LP_m\)
(Corollary~\ref{cor:unlimited-lp}). Our guarantees address finite storage and
common-price implementation while buyers retain their choice of bundles.

For one item of unit capacity, contribution balance weakly dominates rate
balance: an optimal-flow contribution price gives at least \(\LP_1/2\),
which is the exact rate-balanced gain. Across capacities, neither rule
dominates. The following instances have free sellers; buyer rates follow
the listed value order. Their achieved gains give a strict advantage to each
rule (Appendix~\ref{app:rule-comparison}).
\begin{center}
\begin{tabular}{cccccc}
\(C\) & Seller rate & Buyer values & Buyer rates & Rate balance & Contribution balance\\
\hline
\(1\) & \(2\) & \(1\) & \(1\) & \(1/2\) & \(2/3\)\\
\(2\) & \(1\) & \(8,15\) & \(1,1\) & \(10\) & \(69/7\)
\end{tabular}
\end{center}
These gains need not be optimal for the instance.

\paragraph*{Offline with the same inventory capacity.}
Let \(\OPT_{\mathrm{off}}^C\) denote the clairvoyant offline benchmark with
capacity \(C\), starting empty and purchasing each unit before allocating it.
Since \(\OPT_{\mathrm{off}}^C\le\OPT_{\mathrm{off}}^\infty\),
Theorem~\ref{thm:single-main} also gives
\(G_C\ge\frac{C}{C+1}\OPT_{\mathrm{off}}^C\). The multi-item \(1/2\)
guarantees transfer similarly to same-capacity offline.

The unlimited-storage tightness result does not establish an impossibility
for all online policies against same-capacity offline. In the free-seller,
unit-value-buyer example, accepting every feasible trade maximizes sales for
both. With buyer rate one and increasing seller rate, their gains tend to
one, whereas the rate-balanced rule keeps gain \(C/(C+1)\). Its factor thus
remains tight against same-capacity offline, but these examples do not
determine the best guarantee of other online policies against this benchmark.

\paragraph*{Open directions.}
First, finding a polynomial-time construction of the multi-item rate-balanced
prices and their fixed randomized tie-breaking rule remains open
(Theorem~\ref{thm:multi-rate-balance}).

Second, contribution balance applies to arbitrary finite capacities, but its
certificate uses only the statewise fact that every item is either nonempty or
nonfull, and remains at \(1/2\). The single-item theorem
shows that capacity can improve the factor to \(C/(C+1)\).
The unrestricted LP-sampling baseline guarantees
\(C_{\min}/(C_{\min}+1)\) of \(\LP_m\)
(Section~\ref{sec:lp-sampling}).
Obtaining a comparable capacity-sensitive guarantee through static common
posted prices, with XOS buyers choosing utility-maximizing available bundles,
remains open.

Third, a sample-based extension of Feldman, Gravin, and Lucier~\cite{FGL15}
must learn buyer contributions and seller-cost flows while ensuring
feasibility and objective accuracy for the true instance.

Finally, trade is individually rational: sellers receive at least their costs,
and buyers can choose the empty bundle, so seller surplus and buyer utility
are nonnegative. With nonnegative common prices and empty initial inventory,
working capital of \(\sum_jp_jC_j\) funds every transfer
(Appendix~\ref{app:finite-horizon-multi}). Strategic reporting and profit
objectives require further study, as does maximizing gain under a given
capital budget. Profit may require different purchase and sale prices.

\section*{Declaration of AI Assistance}

Generative AI tools provided by OpenAI assisted with literature searches,
drafting the related-work discussion, and revising the exposition. They also
assisted with developing, formalizing, and checking mathematical arguments.
The author reviewed and verified all AI-assisted material and cited sources
and takes responsibility for the final manuscript, including the correctness
of its theorems, algorithms, and proofs.

\clearpage
\bibliographystyle{plainurl}
\bibliography{balanced-posted-prices-arxiv}

\clearpage
\appendix

\section{Finite-horizon arguments}
\label{app:flow-upper-bounds}

Sections~A.1 and~A.2 prove the two flow upper bounds without assuming the
existence of componentwise long-run benchmark flows. Section~A.3 gives the
formal Poisson sampling identity and removes the steady-state assumption used
to present the main multi-item proof.

\subsection{Single-item type-pair flows}

Fix a horizon \(T>0\), a realization of the arrival processes, and an online
or clairvoyant policy starting with empty inventory. Pair every sale by time
\(T\) with one distinct earlier unmatched purchase. Let \(X_{ij}(T)\) be the
number of pairs in which the unit was acquired from seller type \(i\) and
allocated to buyer type \(j\). Pathwise,
\[
\sum_jX_{ij}(T)\le A_i^S(T),
\qquad
\sum_iX_{ij}(T)\le A_j^B(T),
\]
where \(A_i^S(T)\) and \(A_j^B(T)\) are the corresponding seller and buyer
arrival counts.

All seller costs are nonnegative. Deleting the costs of purchases left unsold
at time \(T\), and replacing the gain of every completed pair by its positive
part, can only increase gain. Hence
\[
\mathrm{gain}_{\mathcal A}(0,T)
\le \sum_{i,j}(b_j-a_i)^+X_{ij}(T).
\]
Define \(x_{ij}^T:=\mathbb E[X_{ij}(T)]/T\). Since
\(\mathbb E[A_i^S(T)]=\lambda_iT\) and
\(\mathbb E[A_j^B(T)]=\gamma_jT\), the vector \(x^T\) is feasible
for~\eqref{eq:single-lp}. Therefore
\[
\frac1T\mathbb E[\mathrm{gain}_{\mathcal A}(0,T)]
\le \sum_{i,j}(b_j-a_i)^+x_{ij}^T
\le \LP_1.
\]
Taking the liminf and then the supremum over clairvoyant policies proves
\(\OPT_{\mathrm{off}}^\infty\le\LP_1\).

\subsection{Multi-item bundle and acquisition flows}

Fix a horizon \(T>0\), a realization of the arrival processes, and an online
or clairvoyant policy starting with empty inventory. Let \(Y_{kS}(T)\) be the
number of type-\(k\) buyers allocated bundle \(S\) by time \(T\). For every
allocated unit of item \(j\), match it to one distinct earlier purchase of that
item, and let \(\widehat P_{jr}(T)\) count the matched purchases from seller
type \(r\). Then, pathwise,
\[
\sum_k\sum_{S\ni j}Y_{kS}(T)
=\sum_{r\in R_j}\widehat P_{jr}(T),
\qquad
0\le \widehat P_{jr}(T)\le A^S_{jr}(T),
\]
where \(A^S_{jr}(T)\) is the number of type-\((j,r)\) seller arrivals.
Deleting the costs of all unmatched purchases can only increase gain, so
\[
\mathrm{gain}_{\mathcal A}(0,T)
\le
\sum_{k,S}Y_{kS}(T)v_k(S)
-\sum_{j,r}\widehat P_{jr}(T)a_{jr}.
\]

Set
\[
y_{kS}^T:=\frac{\mathbb E[Y_{kS}(T)]}{T},
\qquad
z_{jr}^T:=\frac{\mathbb E[\widehat P_{jr}(T)]}{T}.
\]
Each type-\(k\) arrival receives at most one bundle, each matched purchase is
a seller arrival of its type, and the displayed pathwise equality gives item
flow balance. Hence \((y^T,z^T)\) is feasible for~\eqref{eq:multi-lp}.
Consequently,
\[
\frac1T\mathbb E[\mathrm{gain}_{\mathcal A}(0,T)]\le\LP_m.
\]
Taking the liminf and then the supremum over clairvoyant policies proves
\(\OPT_{\mathrm{off},m}^\infty\le\LP_m\).

\subsection{Poisson sampling and arbitrary tie-breaking}
\label{app:finite-horizon-multi}

The main text uses steady-state probabilities to display the argument.
For a fixed tie-breaking rule depending only on the current inventory and
fresh independent randomness, the inventory is a finite-state continuous-time
Markov chain and has a limiting distribution. An arbitrary nonanticipatory
rule may depend on time or the full history, so such a distribution need not
exist. The finite-horizon identity below covers both cases.

\begin{samepage}
\begin{lemma}[Finite-horizon Poisson sampling]
\label{lem:poisson-compensation}
Let \(Z(t)\) count a Poisson arrival stream of rate \(r>0\), with arrival
times \(t_1,t_2,\ldots\). Suppose its future increments are independent of
the history available so far, including all past arrivals and independent
policy randomization. For every bounded predictable process \(H(t)\),
determined by that history immediately before time \(t\), and every \(T>0\),
\[
\mathbb E\left[\sum_{\ell:t_\ell\le T}H(t_\ell)\right]
=r\,\mathbb E\left[\int_0^T H(t)\,dt\right].
\]
In particular, if \(\mathbb E[H(t)]\to h\), then
\[
\lim_{T\to\infty}\frac1T
\mathbb E\left[\sum_{\ell:t_\ell\le T}H(t_\ell)\right]=rh.
\]
\end{lemma}
\end{samepage}

\begin{proof}
For \(H(t)=Y\one\{s<t\le u\}\), with bounded \(Y\) determined by
the history at time \(s\), independence of future increments gives
\(\mathbb E[Y(Z(u)-Z(s))]=r(u-s)\mathbb E[Y]\).
Summing proves the identity for simple predictable processes; the standard
bounded predictable approximation gives the general identity. The last
assertion follows by dividing by \(T\) and averaging \(\mathbb E[H(t)]\).
\end{proof}

Taking \(H(t)=\one\{N(t^-)>0\}\) gives the expected number of arrivals
finding stock as \(r\) times the expected time with stock available. Use
\(H(t)=\one\{N(t^-)<C\}\) for free capacity. Inventory may depend
on the sampled stream's past arrivals; only anticipation of future arrivals
is excluded. Zero-rate streams contribute zero. For the balanced single-item
chain, irreducibility implies convergence to its uniform stationary law from
empty inventory. The last assertion of Lemma~\ref{lem:poisson-compensation}
therefore justifies the long-run trade rates used in the single-item proof.

We now prove Theorem~\ref{thm:constructive-multi} without assuming that the
inventory process has a limiting distribution. Fix a nonnegative price vector,
an arbitrary admissible nonanticipatory tie-breaking rule, and a horizon
\(T>0\). Define
\begin{align*}
q_j^B(T)
&:=\frac1T\mathbb E\left[
\int_0^T\one\{N_j(t^-)>0\}\,dt
\right],\\
q_j^S(T)
&:=\frac1T\mathbb E\left[
\int_0^T\one\{N_j(t^-)<C_j\}\,dt
\right].
\end{align*}
The statewise inequality
\[
\one\{N_j>0\}+\one\{N_j<C_j\}\ge1
\]
implies
\begin{equation}
q_j^B(T)+q_j^S(T)\ge1.
\label{eq:access-sum-T}
\end{equation}

Let \(\mathcal U_B(T)\) and \(\mathcal U_S(T)\) be total buyer utility and
seller surplus by time \(T\). Applying the same reference-bundle comparison as
in Lemma~\ref{lem:stationary-decomposition}, followed by the finite-horizon
identity in Lemma~\ref{lem:poisson-compensation} for each buyer stream and
\(H(t)=\one\{N_j(t^-)>0\}\), gives
\[
\mathbb E[\mathcal U_B(T)]
\ge T\sum_jq_j^B(T)B_j(p_j).
\]
Similarly, apply the same identity to each seller stream with
\(H(t)=\one\{N_j(t^-)<C_j\}\). Strict low-cost sellers trade whenever
the relevant inventory is nonfull, while sellers at equality contribute zero
surplus, so
\[
\mathbb E[\mathcal U_S(T)]
=T\sum_jq_j^S(T)L_j(p_j).
\]
Here \(q_j^B(T)\) and \(q_j^S(T)\) are expected occupation-time fractions;
these identities do not require either fraction to converge.

Let \(P_j(T)\) and \(Q_j(T)\) denote the numbers of purchased and allocated
units of item \(j\). Since inventory starts empty,
\[
P_j(T)-Q_j(T)=N_j(T).
\]
The common-price payments therefore satisfy the pathwise identity
\[
\mathrm{gain}_p(0,T)
=\mathcal U_B(T)+\mathcal U_S(T)-\sum_jp_jN_j(T)
\ge
\mathcal U_B(T)+\mathcal U_S(T)-\sum_jp_jC_j.
\]
Combining the preceding displays yields
\begin{equation}
\frac1T\mathbb E[\mathrm{gain}_p(0,T)]
\ge
\sum_j\left(
q_j^B(T)B_j(p_j)+q_j^S(T)L_j(p_j)
\right)
-\frac1T\sum_jp_jC_j.
\label{eq:finite-decomposition}
\end{equation}

For the contribution-balanced prices, let \(d_j=(W_j-A_j)/2\). The itemwise
bounds from the main proof give
\[
B_j(p_j)\ge d_j^+,
\qquad
L_j(p_j)\ge d_j^+.
\]
Using~\eqref{eq:access-sum-T},
\[
q_j^B(T)B_j(p_j)+q_j^S(T)L_j(p_j)
\ge d_j.
\]
Summing over items in~\eqref{eq:finite-decomposition} and using
\eqref{eq:itemwise-lp-master} gives
\[
\mathbb E[\mathrm{gain}_p(0,T)]
\ge
\frac{T}{2}
\left(
\sum_{k,S}y_{kS}v_k(S)-\sum_{j,r}z_{jr}a_{jr}
\right)
-\sum_jp_jC_j.
\]
Dividing by \(T\) and taking the liminf proves
Theorem~\ref{thm:constructive-multi}.

\paragraph*{Working capital.}
For fixed common prices \(p_j\ge0\) and empty initial inventory, cumulative
payments to sellers minus receipts from buyers satisfy, at every time \(t\),
\[
\sum_jp_j\bigl(P_j(t)-Q_j(t)\bigr)
=\sum_jp_jN_j(t)
\le\sum_jp_jC_j.
\]
Thus initial cash of \(\sum_jp_jC_j\) suffices to finance all transfers.
Before a feasible purchase of item \(j\), its free slot leaves at least
\(p_j\) of this cash available. The same inventory accounting yields the
terminal-inventory correction in~\eqref{eq:finite-decomposition}.

\section{Proof of multi-item rate-balance existence}
\label{app:multi-rate-balance}

We prove Theorem~\ref{thm:multi-rate-balance}. The auxiliary random choices
and inventory noise below are used only to obtain a continuous map for a
fixed-point argument. The final policy uses the original arrival streams,
exact utility-maximizing choices, and randomization only at indifference.

\begin{proof}[Proof of Theorem~\ref{thm:multi-rate-balance}]
Delete types with arrival rate zero. For every item, define
\[
\underline a_j:=\min_{r\in R_j}a_{jr},
\qquad
\overline v_j:=\max_{k\in K}v_k(\{j\}),
\qquad
\mathcal J:=\{j\in M:\underline a_j<\overline v_j\}.
\]
Positive total supply for each item and positive total buyer rate make these
minima and maxima well-defined.

\paragraph*{Restricting an optimal flow.}
XOS valuations satisfy the marginal bound
\[
v_k(S)-v_k(S\setminus\{j\})\le v_k(\{j\})
\le\overline v_j.
\]
To see the first inequality, take a supporting clause for \(S\): its value
on \(S\setminus\{j\}\) is at most \(v_k(S\setminus\{j\})\), and its
weight on \(j\) is at most \(v_k(\{j\})\).
If \(j\notin\mathcal J\), remove \(j\) from every bundle in a feasible
flow solution, merging variables for identical resulting bundles, and set
all its purchase variables to zero. Buyer-rate constraints and every other
item's flow remain unchanged. The lost buyer value is at most
\(x_j\overline v_j\), whereas the saved seller cost is at least
\(x_j\underline a_j\). Applying this operation successively to an optimal
solution preserves optimality and leaves positive flow only on
\(\mathcal J\).

For items outside \(\mathcal J\), set the price to zero and reject sellers
at equality. These inventories remain empty. If \(\mathcal J=\varnothing\),
the restricted optimal solution has value zero and this policy proves the
theorem. Otherwise, work on the finite active inventory state space
\[
\mathcal N:=\prod_{j\in\mathcal J}\{0,1,\ldots,C_j\},
\qquad A(n):=\{j\in\mathcal J:n_j>0\}.
\]

\paragraph*{A continuous auxiliary inventory process.}
Choose \(H\) greater by at least one than every active seller cost and
every \(\overline v_j\), and consider prices in \([0,H]^{\mathcal J}\).
For \(\varepsilon>0\), a seller of type \((j,r)\) is accepted, whenever
capacity is available, with probability
\[
\alpha_{jr}^{\varepsilon}(p_j)
:=\min\left\{1,\max\left\{0,\frac{p_j-a_{jr}}{\varepsilon}\right\}\right\}.
\]
Given available set \(A\), a buyer of type \(k\) chooses \(D\subseteq A\)
with probability
\[
\beta_{k,A,D}^{\varepsilon}(p)
:=\frac{\exp\bigl((v_k(D)-\sum_{j\in D}p_j)/\varepsilon\bigr)}
{\displaystyle\sum_{E\subseteq A}
\exp\bigl((v_k(E)-\sum_{j\in E}p_j)/\varepsilon\bigr)}.
\]
In addition, give each active inventory coordinate independent upward and
downward moves at rate \(\varepsilon\), whenever the corresponding move is
feasible. These extra moves are fictitious and will vanish in the limit.
They make the auxiliary continuous-time Markov chain irreducible. Its
transition rates are continuous in \(p\), so its unique stationary
distribution \(\pi_p^{\varepsilon}\) is continuous in \(p\). This follows
also by solving its stationary linear equations with normalization: the
unique solution varies continuously with the matrix entries.

Write
\[
u_j^{\varepsilon}(p):=
\sum_r\lambda_{jr}\alpha_{jr}^{\varepsilon}(p_j),
\qquad
d_j^{\varepsilon}(n,p):=
\sum_k\gamma_k\sum_{\substack{D\subseteq A(n)\\j\in D}}
\beta_{k,A(n),D}^{\varepsilon}(p).
\]
These are the potential purchase intensity before checking capacity and the
buyer consumption intensity in state \(n\), excluding fictitious moves. Let
\(\pi_j^{\varepsilon}(\ell):=\Pr_{\pi_p^{\varepsilon}}[N_j=\ell]\).
For \(0\le\ell<C_j\), equality of stationary rates crossing the inventory
level between \(\ell\) and \(\ell+1\) gives
\begin{equation}
(u_j^{\varepsilon}(p)+\varepsilon)\pi_j^{\varepsilon}(\ell)
=\mathbb E_{\pi_p^{\varepsilon}}\left[
\bigl(d_j^{\varepsilon}(N,p)+\varepsilon\bigr)
\one\{N_j=\ell+1\}\right].
\label{eq:rate-balance-crossing}
\end{equation}
This uses only that each purchase or consumption changes \(N_j\) by one;
the other inventory coordinates may be coupled.

\paragraph*{Opposite signs on the price faces.}
Recall \(\Gamma=\sum_k\gamma_k>0\) and
\(\Lambda_j=\sum_r\lambda_{jr}>0\). Choose \(\varepsilon\) sufficiently
small that
\[
0<\varepsilon<1,
\qquad \frac{2\varepsilon}{\Gamma}<1,
\qquad
\Gamma e^{-1/\varepsilon}+\varepsilon
<\min_{j\in\mathcal J}\Lambda_j.
\]
If \(p_j=0\), then \(u_j^{\varepsilon}=0\). Whenever \(j\) is available,
pair each bundle \(D\) excluding \(j\) with \(D\cup\{j\}\). Monotonicity
and the zero price make the latter's utility at least the former's, so its
exponential weight is at least as large. Every buyer therefore includes
\(j\) with probability at least \(1/2\), giving
\(d_j^{\varepsilon}(n,p)\ge\Gamma/2\) whenever \(n_j>0\).
Equation~\eqref{eq:rate-balance-crossing} implies
\[
\pi_j^{\varepsilon}(\ell+1)
\le\frac{2\varepsilon}{\Gamma}\pi_j^{\varepsilon}(\ell).
\]
Thus \(\pi_j^{\varepsilon}(C_j)<\pi_j^{\varepsilon}(0)\).

If \(p_j=H\), every seller of item \(j\) is accepted, so
\(u_j^{\varepsilon}=\Lambda_j\). For
\(D\subseteq A(n)\setminus\{j\}\), the marginal bound above gives
\[
v_k(D\cup\{j\})-\sum_{h\in D\cup\{j\}}p_h
\le v_k(D)-\sum_{h\in D}p_h+\overline v_j-H.
\]
The paired exponential weights now show that the inclusion probability is
at most \(e^{-(H-\overline v_j)/\varepsilon}\allowbreak\le e^{-1/\varepsilon}\).
Using~\eqref{eq:rate-balance-crossing},
\[
\pi_j^{\varepsilon}(\ell)
\le\frac{\Gamma e^{-1/\varepsilon}+\varepsilon}{\Lambda_j}
\pi_j^{\varepsilon}(\ell+1).
\]
Hence \(\pi_j^{\varepsilon}(C_j)>\pi_j^{\varepsilon}(0)\). Both face
inequalities are strict because irreducibility makes every inventory-level
probability positive. They hold uniformly over all the other prices.

\paragraph*{The fixed point and its limit.}
Define the continuous imbalance
\[
F_j^{\varepsilon}(p):=
\Pr_{\pi_p^{\varepsilon}}[N_j=C_j]
-\Pr_{\pi_p^{\varepsilon}}[N_j=0]
\]
and the continuous map of the price cube to itself
\[
T_j^{\varepsilon}(p):=
\min\{H,\max\{0,p_j-F_j^{\varepsilon}(p)\}\}.
\]
Brouwer's fixed-point theorem~\cite[Theorem~2.1]{AGR18} supplies a fixed
point \(p^{\varepsilon}\). The face signs exclude a fixed point with any
coordinate equal to zero or \(H\). At an interior fixed point the displayed
equation forces \(F_j^{\varepsilon}(p^{\varepsilon})=0\) for every active
item.

Take a sequence \(\varepsilon\downarrow0\). Compactness allows a common
subsequence along which the prices, all seller acceptance probabilities,
all buyer-choice probabilities for every type and available set, and the
stationary distributions converge. Denote their limits by \(p^*\),
\(\alpha_{jr}^*\), \(\beta_{k,A,D}^*\), and \(\pi^*\).
If \(p_j^*>a_{jr}\), the seller limit is one; if \(p_j^*<a_{jr}\), it is
zero. Intermediate acceptance probabilities can occur only at equality.
Likewise, any bundle with a strict utility gap from the maximum at \(p^*\)
has an exponentially vanishing choice probability. Thus \(\beta^*\) is
supported only on exact utility maximizers. These limits define a fixed
stationary randomized policy: seller coins depend only on seller type,
and buyer probabilities only on buyer type and the available set.

Let \(Q^*\) be the limiting policy's transition-rate matrix (generator),
with diagonal entries making each row sum to zero. The auxiliary
transition-rate matrices converge to \(Q^*\) as the fictitious moves vanish.
Taking limits in the stationary
equations gives \(\pi^*Q^*=0\), and endpoint balance gives
\begin{equation}
\Pr_{\pi^*}[N_j=0]=\Pr_{\pi^*}[N_j=C_j]
\quad\text{for all }j\in\mathcal J.
\label{eq:multi-state-endpoints}
\end{equation}

\paragraph*{Convergence from empty inventory.}
It remains to show that this stationary distribution is the one reached
from the required initial state. A \emph{closed recurrent class} is a set
of inventory states that the process cannot leave and within which every
state can be reached from every other. The potential purchase intensity under the
limiting policy is the state-independent number
\[
u_j^*:=\sum_r\lambda_{jr}\alpha_{jr}^*.
\]
Suppose \(u_j^*=0\) for some \(j\in\mathcal J\). Strictly acceptable
sellers would give positive intensity, so \(p_j^*\le\underline a_j\).
Choose a positive-rate buyer type \(k\) with
\(v_k(\{j\})=\overline v_j>\underline a_j\). While \(j\) is stocked,
this buyer has a bundle with strictly positive utility and hence must
choose a nonempty bundle. A sequence of at most
\(\sum_{h\in\mathcal J}C_h\) arrivals of this buyer without seller
arrivals must consume \(j\): as long as it remains stocked, each arrival
decreases total inventory by at least one. Such a sequence has positive
probability from every state with \(j\) stocked. Since \(j\) can never
be replenished, no closed recurrent class can retain it. Every stationary
distribution would therefore have \(N_j=0\) almost surely, contradicting
\eqref{eq:multi-state-endpoints}. Thus \(u_j^*>0\) for every active item.

From any active inventory state, a sequence of accepted seller arrivals
and no buyers can now reach the full active inventory vector. It has
positive probability because every \(u_j^*>0\). The finite chain therefore
has exactly one closed recurrent class, containing that full vector. Its
unique stationary distribution is \(\pi^*\), and the distribution starting
empty converges to it. Items outside \(\mathcal J\) remain empty under
the previously specified zero-price rule.

\paragraph*{The welfare guarantee.}
The events \(N_j=0\) and \(N_j=C_j\) are disjoint. Their common probability
in~\eqref{eq:multi-state-endpoints} is at most \(1/2\), so
\[
q_j^B=q_j^S\ge\frac12\qquad(j\in\mathcal J).
\]
The optimal flow chosen at the start has positive flow only on
\(\mathcal J\). Apply Corollary~\ref{cor:half-availability} to this flow
and the limiting policy. Convergence from empty also justifies its
long-run gain: in~\eqref{eq:finite-decomposition}, the expected occupation
fractions converge to \(q_j^B,q_j^S\), and the bounded terminal-inventory
term vanishes. The same itemwise bounds used in the corollary therefore give
\(G(p^*)\ge\LP_m/2\). The offline upper bound finishes the proof.
\end{proof}

\section{Exact oracle model and primal recovery}
\label{app:oracle-details}

We use the standard Turing model for rational data. Arrival rates and seller
costs are rational numbers of polynomial encoding length. Exact demand and
XOS oracles run in time polynomial in the instance and query encoding lengths.
A demand query at rational nonnegative prices returns a utility-maximizing
bundle. An XOS query returns a supporting additive clause with
rational weights whose encoding length is polynomial in the instance length,
uniformly over all queried bundles. Types of rate zero may be deleted before
running the algorithm; this does not change the LP or the posted-price policy
and can only decrease the support bound.
A demand query restricted to the available set \(A\) can be simulated by
assigning every unavailable item price \(v_k(M)+1\), where \(v_k(M)\) is
obtained from an XOS query on \(M\).

The dual feasible region in~\eqref{eq:multi-dual-buyer} and
\eqref{eq:multi-dual-seller} has dimension
\[
d=|K|+m+\sum_j|R_j|.
\]
Its explicit nonnegativity and seller inequalities have polynomial encoding
length. For a proposed rational point \((\alpha,\mu,\eta)\), these inequalities
are checked directly. For each buyer type, the demand and XOS queries described
in the proof of Theorem~\ref{thm:computability} either return a genuinely
violated buyer inequality or certify all buyer inequalities of that type. The
returned inequality has coefficients in \(\{0,1\}\), and its right-hand side is
the rational value of the returned supporting clause. Thus the dual is a
well-described rational polyhedron with an exact strong separation oracle whose
outputs have polynomial encoding length.

The zero solution is feasible for the primal, and the rate constraints bound
all primal variables, so the primal has a finite attained optimum. Linear
programming duality gives the same conclusion for the dual. We apply the strong optimization theorem and oracle-inequality dual-recovery
lemma of Gr\"otschel, Lov\'asz, and
Schrijver~\cite[Theorem~6.4.9 and Lemma~6.5.15]{GLS88} directly to this dual
polyhedron, without intersecting it with an artificial bounding box. Hence the
recovered optimal multipliers are supported only on genuine nonnegativity, seller,
and buyer inequalities. In particular, every positive
multiplier on a buyer inequality identifies an actual bundle \(S\) returned by
the separation oracle.

It remains only to account for the restriction \(\mu\ge0\). The unrestricted
flow-dual variables may be replaced by their positive parts without increasing
the dual objective or violating any constraint, as observed in the main proof.
The dual of the formulation with \(\mu\ge0\) has the relaxed item-flow
constraints
\[
\sum_k\sum_{S\ni j}y_{kS}\le\sum_r z_{jr}\qquad\forall j.
\]
Given the recovered primal solution, reduce the \(z_{jr}\) values itemwise until equality
holds. Such a reduction is always possible, preserves all seller upper bounds,
and cannot decrease the objective because \(a_{jr}\ge0\). The resulting
solution is feasible for~\eqref{eq:multi-lp} and therefore has value exactly
\(\LP_m\). The trimming introduces no positive variable.

Before trimming, the recovered multiplier solution is basic, so at most
\(d\) dual inequalities have positive multipliers. These multipliers
include the variables \(y,z\) and the primal slacks; consequently at most \(d\)
of the recovered \(y,z\) variables are positive. XOS queries on the positive
\(y_{kS}\) variables return all clauses needed to compute \(W_j\), while
\(x_j,A_j\), and the rational prices in~\eqref{eq:constructive-price} are then
obtained by polynomially many exact arithmetic operations.

\section{Extended related work}
\label{app:extended-related}

\paragraph*{Classical prophet inequalities and posted prices.}
The classical prophet-inequality literature compares an online stopping rule
with a prophet who sees all realizations~\cite{KS77,KS78}. Samuel-Cahn's
median rule~\cite{SC84} and the rank-one specialization of the balanced-price
argument of Kleinberg and Weinberg~\cite{KW12} provide the two proof templates
used in this paper. Later work developed prophet inequalities and posted-price
mechanisms under richer feasibility constraints; representative examples
include sequential posted pricing~\cite{CHMS10}, online contention
resolution~\cite{FSZ16}, and the pricing framework of D\"utting et
al.~\cite{DFKL20}; Correa et al.~\cite{CFHOV19} survey this literature.
For combinatorial auctions with XOS buyers, Feldman, Gravin, and
Lucier~\cite{FGL15} obtain static item prices from supporting additive clauses.
Braun and Kesselheim~\cite{BK23SOSA} give a simplified prophet-inequality
treatment of combinatorial auctions, while Banihashem et
al.~\cite{BHHKO24} study the power and limitations of posted prices in prophet
settings. More recent combinatorial results address limited samples
and multiple units~\cite{DKLRS24,CDHW25}. These papers concern one-sided
allocation: the present model adds recurring procurement from cost-bearing sellers and
an inventory process generated by the posted-price policy.

\paragraph*{Stationary prophets, uncertain supply, and dynamic matching.}
Uncertain-supply prophet inequalities and stochastic matching already combine
online demand with availability over time~\cite{ABGMW20,AS22,CILLN20}.
Kessel et al.~\cite{KSSW22} isolate the stationary-prophet model, in which
copies arrive exogenously, remain usable for random lifetimes, and must be
allocated to recurring buyers. Patel and Wajc~\cite{PW24} extend this line to
combinatorial feasibility constraints and submodular objectives. They also
identify XOS valuations and posted-price mechanisms as natural further
directions. Our contribution-balanced result addresses a two-sided variant of
those directions with nonperishing inventory. Seller opportunities arrive
exogenously, but inventory is not
replenished automatically: every stored unit must be purchased from a seller
at the common posted price, and the seller's cost enters social welfare. Buyers
have XOS preferences over bundles, and their demand can correlate the item
inventories. The contribution-balanced proof therefore avoids both
independence assumptions and any need to prescribe the stationary
availabilities. The separate rate-balance result controls stock and
free-capacity probabilities with jointly chosen prices and randomized
tie-breaking.

\paragraph*{Trading prophets.}
The trading-prophet model replaces a single stopping decision with repeated
buying and selling over a finite observed sequence~\cite{CCDHOS23}. Subsequent
work considers several assets and initial capital~\cite{RCV25,ABLV26}.
Aggarwal et al.~\cite{AGWZ26} study an
asymmetric model with separate buying and selling opportunities and prove an
impossibility without initial inventory, even when the two sides have
identical distributions. Our model is also asymmetric and begins empty, but
it has recurring stationary streams, held inventory does not perish, and
there is no terminal liquidation requirement. The objective also differs:
we maximize long-run social gain, whereas trading prophets maximize the
trader's finite-horizon profit. These distinctions allow balanced long-run
flows to overcome the initial transient.

\paragraph*{Strategic two-sided combinatorial markets.}
Colini-Baldeschi et al.~\cite{CGKLRT20} give incentive-compatible and
budget-balanced constant-factor mechanisms for two-sided combinatorial
auctions with XOS buyers. Braun and Kesselheim~\cite{BK23} use
prophet-inequality ideas to obtain simple strategic mechanisms, including a
\(1/2\) welfare approximation for combinatorial double auctions with
unit-supply sellers and XOS buyers. D\"utting et al.~\cite{DFLLR26} show that
one sample from each seller distribution suffices for broad classes of
approximately efficient two-sided mechanisms and give black-box reductions
from one-sided mechanisms.

These papers study one-shot reallocation markets in which sellers initially
own the items and incentive compatibility, individual rationality, and budget
balance are central. Our agents are price-taking; seller arrivals instead
create recurring procurement opportunities, every unit sold must have been
purchased earlier, and performance is measured by long-run gains from trade
against a clairvoyant unlimited-inventory benchmark.

\paragraph*{Gains from trade and common fixed prices.}
Brustle et al.~\cite{BCWZ17} obtain a \(1/2\) approximation to the best
incentive-compatible gains-from-trade benchmark in double auctions with
unit-demand buyers and general downward-closed feasibility constraints.
Colini-Baldeschi et al.~\cite{CGKLT17} study the gains from trade achievable
by common fixed prices in bilateral trade and double auctions. These results
further motivate common prices, but concern one-shot strategic exchange among
initially endowed agents. In our recurrent model, the price also controls
procurement and thereby the future inventory state.

\paragraph*{Online intermediation.}
Online intermediation studies incentive compatibility, individual rationality,
and budget balance~\cite{GKL17}. Our agents take prices as given; the welfare
guarantees omit strategic reporting, intermediary profit, and operating costs.
The transfers require only the bounded working capital established in
Appendix~\ref{app:finite-horizon-multi}. Strategic recurring two-sided
mechanisms and optimizing gain under a given capital budget require further analysis.

\section{Single-item implementation details}
\label{app:single-implementation}

\begin{remark}[Indifference and implementation]
\label{rem:rate-balance-implementation}
The price \(\tau\) is an actual payment, but the rate-balanced
guarantee uses the prescribed randomization when a cost or value equals it.
If indifferent sellers and buyers instead follow a different rule, the equality
\(u=v\) need not hold. In contrast, the contribution-balanced multi-item guarantee in
Theorem~\ref{thm:constructive-multi} is valid under any nonanticipatory
tie-breaking rule, regardless of how it resolves seller indifference and buyer
demand ties.
\end{remark}

\section{Single-item rule comparisons}
\label{app:single-comparison}
\label{app:rule-comparison}

For one item, the matching and bundle-flow relaxations have the same value.
Given a feasible matching in~\eqref{eq:single-lp}, delete pairs with
\(b_j<a_i\), which have zero positive-part gain. Its buyer and seller
marginal rates then give a feasible one-item flow in~\eqref{eq:multi-lp}
with the same objective. Conversely, couple the item-bearing buyer and
seller rates of any feasible one-item flow into a matching with those
marginals. Since \((b_j-a_i)^+\ge b_j-a_i\), its matching objective is at
least the flow objective. Thus the LP values agree. At unit capacity,
Theorems~\ref{thm:single-main} and~\ref{thm:constructive-multi} give exact
gain \(\LP_1/2\) for rate balance and at least that value for contribution
balance from an optimal flow.

All sellers in the examples below have cost zero. The fixed rules induce
irreducible finite birth-death chains, which converge from empty inventory;
PASTA therefore justifies the long-run stock calculations. For positive
potential purchase and sale rates \(u,v\), neighboring inventory levels
satisfy \(\pi_{n+1}/\pi_n=u/v\).

\paragraph*{Contribution balance gives larger gain.}
Let \(C=1\), seller rate \(2\), and value-\(1\) buyer rate \(1\).
The LP value is \(1\), attained by purchasing and allocating at rate \(1\).
Use rate-balanced price \(\tau=0\), independently accepting each
indifferent seller with probability \(1/2\). All buyers strictly accept,
so \(u=v=1\), and the gain is \(1/2\).

The optimal flow has \(x=1\), \(W=1\), and \(A=0\), giving contribution
price \(p=1/2\). All sellers and buyers strictly accept, so \(u=2\) and
\(v=1\). The stationary probabilities are \((\pi_0,\pi_1)=(1,2)/3\).
Buyers find stock with probability \(2/3\), giving gain \(2/3\).

\paragraph*{Rate balance gives larger gain.}
Let \(C=2\), seller rate \(1\), and buyer values \(8\) and \(15\), each
with arrival rate \(1\). The LP value is \(15\): at most one unit can be
allocated per unit time, and allocating only to value-\(15\) buyers attains
this upper bound. At rate-balanced price \(\tau=10\), sellers and the
value-\(15\) buyers strictly accept, while value-\(8\) buyers reject.
Thus \(u=v=1\), inventory is uniform on \(\{0,1,2\}\), and gain is
\(15\cdot(2/3)=10\).

The optimal flow has \(x=1\), \(W=15\), and \(A=0\), giving contribution
price \(p=15/2\). Both buyer types strictly accept, so \(u=1\) and
\(v=2\). The stationary probabilities are
\((\pi_0,\pi_1,\pi_2)=(4,2,1)/7\). Hence buyers find stock with probability
\(3/7\), and the gain is
\[
(8+15)\frac37=\frac{69}{7}<10.
\]

\end{document}